\documentclass{llncs}
\usepackage{txfonts}
\usepackage{euscript}
\usepackage{amssymb}
\usepackage[all]{xy}
\usepackage{url}

\newcommand{\<}{\sqsubseteq}
\newcommand{\com}{\mathop{\circ}}

\newcommand{\lra}{\longrightarrow}

\newcommand{\ps}{\mathcal{P}}
\newcommand{\Rel}{\textup{\textbf{Rel}}}
\newcommand{\rel}[1]{\mathrm{Rel}(#1)}
\newcommand{\Sets}{\textup{\textbf{Sets}}}
\newcommand{\B}{\mathop{_{Act^r}\mathcal{S}_{Act^l}}}   
\newcommand{\Ord}{\mathop{_{Act^r}\hspace{-1.2mm}\sqsubseteq_{Act^l}}} 
 
\newcommand{\CSO}{\sqsubseteq^{\mathit{Conf}}} 
\newcommand{\relord}[2]{\mathrm{Rel}_{#2}(#1)}
\newcommand{\CSEmpty}{\sqsubseteq^{\mathit{C\emptyset}}}
\newcommand{\CSNotEmpty}{\sqsubseteq^{\mathit{C\neg\emptyset}}}
\newcommand{\CSEmptya}{\sqsubseteq^{\mathit{a,\emptyset}}}
\newcommand{\CSNotEmptya}{\sqsubseteq^{\mathit{a,\neg\emptyset}}}

\begin{document}

\title{Non-Strongly Stable Orders Also Define Interesting Simulation
Relations\thanks{Research supported by the Spanish
projects DESAFIOS TIN2006-15660-C02-01, WEST TIN2006-15578-C02-01, PROMESAS
S-0505/TIC/0407 and UCM-BSCH GR58/08/910606.}}

\author{Ignacio F\'abregas \and David de Frutos Escrig
\and Miguel Palomino}
\institute{Departamento de Sistemas Inform\'aticos y
Computaci\'on,  UCM\\
\email{fabregas@fdi.ucm.es \quad \{miguelpt, defrutos\}@sip.ucm.es}}

\maketitle

\begin{abstract}
We present a study of the notion of coalgebraic simulation introduced by Hughes
and Jacobs. Although in their original paper they allow any functorial order in
their definition of coalgebraic simulation, for the simulation relations to have
good properties they focus their attention on functors with orders which
are strongly stable. This guarantees a so-called ``composition-preserving''
property from which all the desired good properties follow.
We have noticed that the notion of strong stability not only ensures such good
properties but also ``distinguishes the direction'' of the simulation.
For example, the classic notion of simulation for labeled transition systems,
the relation ``$p$ is simulated by $q$'', can be defined as a coalgebraic
simulation relation by means of a strongly stable order, whereas the opposite
relation, ``$p$ simulates $q$'', cannot.
Our study was motivated by some interesting classes of simulations
that illustrate the application of these results: covariant-contravariant
simulations and conformance simulations.
\end{abstract}

\section{Introduction and presentation of our new results}

Simulations are a very natural way to compare systems defined by
transition systems or other related mechanisms based on the description of
systems by means of the  actions they can execute at each of their states
\cite{Park81}. They can be enriched in several ways to obtain, in particular,
the important ready simulation semantics \cite{Bloom95,LarsenSkou91}, as
well as other more elaborated ones such as nested simulations \cite{Groote92}.
Quite recently we have studied the general concept of constrained simulation
\cite{DeFrutosGregorio08}, proving that all the simulation relations constrained
by an adequate condition have similar properties. The semantics of these constrained
simulations is also the basis for our unified
presentation of the
semantics of processes \cite{DeFrutosEtAl08c}, where all the semantics in the
ltbt-spectrum \cite{VanGlabbeek01} (and other new semantics) are classified in a
systematic way.

Hughes and Jacobs \cite{HughesJacobs04} have also developed a
systematic study of simulation-like relations, this time in a purely coalgebraic
context, so that simulations are studied in connection with bisimulations
\cite{Park81}, the fundamental concept to define equivalence in the coalgebraic
world. Their coalgebraic simulations are defined in terms of an order $\<$
associated to the functor $F$ corresponding to the coalgebra $c:X\lra FX$ that
we want to observe.
In this way they obtain a very general notion of coalgebraic simulation, not
only because all functors $F$ are considered, including in particular the
important class of polynomial functors, but also because by changing the family
of orders $\<_X$ many different families of simulation relations can be
obtained.
The general properties of these simulations can be studied in the defined
coalgebraic framework, thus avoiding the need of similar
proofs for each of the particular classes of simulations.

Certainly, this generic presentation of the notion of coalgebraic simulation has
as advantage that it provides a wide and abstract framework where one can try to
isolate and take advantage of the main properties of all the simulation-like
relations. However, at the same time it can be argued that
the proposal fails to capture in a tight manner the spirit of simulation
relations because, in addition to the natural notions of simulations, the
framework also allows for other less interesting relations. This has as a result
that some natural properties of simulations cannot be proved in general,
simply due to the fact that they are not satisfied by all of the permitted
coalgebraic simulation relations.
For instance, the induced similarity relation between systems is not always an
order because transitivity is not always satisfied. In order to guarantee
transivity,
and other related properties of coalgebraic simulations, Jacobs and Hughes
introduce in \cite{Jacobs03} the composition-preserving property to the order
$\<$ that induces the simulation relation. In \cite{HughesJacobs04} they
continue with the study of the topic and present \emph{stability} of orders as
a natural categorical property to guarantee that an order is
composition-preserving. They also
comment that stability is not easy to check and introduce a stronger
condition (that we will call right-stability) so that, whenever applicable, the
checking of the main properties of coalgebraic simulations becomes much simpler
than in the general case.

Roughly speaking, given an order $\<_X$ on $FX$ for each set $X$, the induced
coalgebraic simulations are defined in the same way as bisimulations for $F$,
but allowing a double application of $\<$ on the two sides of the defined
relation. More precisely, instead of the functor $\rel{F}$ defining plain
bisimulations, $\relord{F}{\<}$ defined as $\<_Y\com\rel{F}\com \<_X$ is used.
There are several interesting facts hidden behind the apparent simplicity of
this definition. The first one is that, in general, it only defines an order
and not an equivalence relation, even if it is based on bisimulations (that
always define an equivalence relation, namely, bisimilarity).
The reason is that the order $\<$ appears ``in the same direction'' on both
sides of the definition, thus breaking its symmetry.
However, we can also define some equivalence relations weaker than bisimilarity
by using an equivalence relation $\equiv$ as the order $\<$. Another
interesting fact is that whenever we define a coalgebraic simulation by
using $\<$, the inverse order $\sqsupseteq$
defines the inverse relation of that defined by $\<$ once we also interchange
the roles of the related sets $X$ and $Y$ (so we could say that we are defining
in fact the same relation but looking at it from the other side). Stability is
also a symmetric condition, so that whenever an order $\<$ on a functor $F$ is
stable, the inverse order $\sqsupseteq$ is stable for $F$, too. This is quite
reasonable, since stability is imposed in order to guaratee transitivity of the
generated similarity relation and the inverse of a transitive relation
is also transitive, so that whenever $\<$ generates an ``admissible''
similarity relation (meaning that it is an order), the inverse order
$\sqsupseteq$ must be also admissible.

It is worth noting that the stronger condition guaranteeing stability
is asymmetric. In fact, Hughes and Jacobs prove in
\cite{HughesJacobs04} that ``right-stability'' implies that
\begin{equation}\label{cond1}
{\rel{F}(R)\com\<_X}\subseteq{\<_Y\com\rel{F}(R)},
\end{equation}
which in fact motivates our name for the condition.

A second surprise was to notice that, in most cases, right-stability
induces a ``natural direction'' on the orders defining the coalgebraic
simulation. For instance, for plain similarity over labeled transition systems,
the inclusion order $\subseteq$ induces the classic simulation
relation while the reversed inclusion $\supseteq$ induces the opposite
``simulated by'' relation: the first one is
right-stable while the second is not.

All these general results arose when
trying to integrate two new simulation-like notions as coalgebraic
simulations definable by a stable order, so that we could obtain for free all
the good properties that have been proved in \cite{HughesJacobs04} for
this class of relations.

The first new simulation notion is that of covariant-contravariant
simulations, where the alphabet of actions $Act$ is partitioned into
three disjoint sets $\textit{Act}^l$, $\textit{Act}^r$, and
$\textit{Act}^{\mathit{bi}}$.
The intention is for the simulation to treat the actions in $\textit{Act}^l$
like in the ordinary case, to interchange the role of the related processes for
those actions in $\textit{Act}^r$, and to impose
a symmetric condition like that defining bisimulation for the actions in
$\textit{Act}^{\mathit{bi}}$.

The second notion, conformance simulations, captures the conformance relations
\cite{Leduc92,Tretmans96} that several authors introduced in order to formalize
the notion of possible implementations.
Like covariant-contravariant simulations, they can be defined as coalgebraic
simulations for some stable order which is not right-stable neither
left-stable.
We show that the good properties of these two classes of orders are preserved
in those orders that can be seen as a kind of composition of right-stable and
left-stable orders. We use this fact to derive the stability of the
orders defining both covariant-contravariant and conformance simulations.

\section{Coalgebraic simulations and stability}

Given a category $\mathbb{C}$ and an endofunctor $F$ in $\mathbb{C}$, an
$F$-coalgebra, or just a coalgebra, consists of an object $X\in\mathbb{C}$
together with a morphism $c:X\lra FX$. We often call $X$ the state space and
$c$ the transition or coalgebra structure.

An arbitrary endofunctor $F:\Sets \lra \Sets$ can be lifted to a functor
in the category $\Rel$ over $\Sets\times \Sets$ of relations, $\rel{F} : \Rel
\lra \Rel$.
In set-theoretic terms, for a relation $R\subseteq X_1 \times X_2$,
\[
\rel{F}(R) = \{ \langle u,v \rangle \in FX_1 \times FX_2 \mid
                \exists w\in F(R).\, F(r_1)(w) = u, F(r_2)(w) = v \}.
\]

A \emph{bisimulation} for coalgebras $c : X\lra FX$ and $d:Y \lra FY$ is a
relation
$R\subseteq X\times Y$ which is ``closed under $c$ and $d$'':
\[
\textrm{if $(x,y) \in R$ then $(c(x), d(y)) \in \rel{F}(R)$},
\]

\noindent where the $r_i$ are the projections of $R$ into
$X$ and $Y$. Sometimes we shall use the term $F$-bisimulation to emphasize the
functor we are working with.

Bisimulations can also be characterized by means of spans, using the general
categorical definition by Aczel and Mendler~\cite{AczelMendler89}:
\[
\xymatrix@R=5.0ex{
 {X}\ar[d]_{c}   & {R} \ar[d]_{e}\ar[l]_{r_1}\ar[r]^{r_2}    &
Y\ar[d]_{d} \\
 {FX}           & {FR} \ar[l]_{Fr_1}\ar[r]^{Fr_2}           & {FY}
}
\]
$R$ is a bisimulation iff it is the carrier of some coalgebra $e$ making
the above diagram commute. Alternatively, bisimulations can also be defined as
the
$\rel{F}$-coalgebras in the category $\Rel$.

We will also need the general concept of simulation introduced by Hughes and
Jacobs~\cite{HughesJacobs04} using orders on functors.
Let $F : \Sets\lra\Sets$ be a functor.
\emph{An order on $F$} is defined by means of a functorial collection of
preorders
$\sqsubseteq_X \subseteq FX\times FX$ that must be preserved by renaming:
for every $f: X\lra Y$, if $u\sqsubseteq_X u'$ then
${Ff(u)}\sqsubseteq_Y{Ff(u')}$.

Given an order $\sqsubseteq$ on $F$, a \emph{$\sqsubseteq$-simulation} for
coalgebras $c: X\lra FX$ and $d: Y\lra FY$ is a relation $R\subseteq X\times Y$
such that
\[
\textrm{if $(x,y) \in R$ then $(c(x), d(y)) \in \relord{F}{\<}(R)$},
\]
where the lax relation lifting
$\relord{F}{\<}(R)$ is
$\sqsubseteq_Y\com\rel{F}(R)\com\sqsubseteq_X$,
which can be expanded to
\[
\relord{F}{\<}(R)=\{(u,v)\mid\exists w\in F(R).\; u\sqsubseteq_X
Fr_1(w)\wedge Fr_2(w)\sqsubseteq_Y v\}.
\]

\noindent Alternatively, $\<$-simulations are just the
$\relord{F}{\<}$-coalgebras in $\Rel$.

Sometimes, when $f:X\lra Y$ and $A\subseteq X$ we will simply write
$f(A)$ for the image $\coprod_f (A)$.

A functor with order $\<$ is \emph{stable} \cite{HughesJacobs04} if the relation
lifting $\relord{F}{\<}$ commutes with
substitution, that is, if for every $f:X\lra Z$ and $g:Y\lra W$,
$\relord{F}{\<}((f\times g)^{-1}(R))=(Ff\times
Fg)^{-1}(\relord{F}{\<}(R))$.\footnote{In fact, the
inclusion $\subseteq$ always holds.} They also define a stronger condition that
we are going to call right-stability.

\begin{definition}[\cite{HughesJacobs04}]
We will say that a functor $F$ with order $\<$ is \textbf{right-stable} if, for
every function $f:X\lra Y$, we have\footnote{Again, the other inclusion is
always true since $\<$ functorial means that $Ff(u)\<_Y Ff(v)$ if $u\<_X v$.}
\begin{equation}\label{strong}
{(id\times Ff)^{-1}\<_Y}\;\subseteq\; \coprod_{Ff\times id}\<_X.
\end{equation}
\end{definition}

According to \cite{HughesJacobs04}, condition (\ref{strong}) is equivalent to
(a) $F$ being stable and (b) for every relation $R\subseteq X\times Y$,
\begin{equation}\label{cond-right}
{\rel{F}(R)\com\<_X}\;\subseteq\;{\<_Y\com\rel{F}(R)}.
\end{equation}

Right-stability was introduced by arguing that it is easier to check
than plain stability, while being satisfied by nearly
all orders discussed in that paper. Surprisingly, one cannot find in
\cite{HughesJacobs04} a clear explanation of the reason why right-stable orders
are easier to analyze. In our opinion, the crucial fact is that from
(\ref{cond-right}) we can immediately conclude that
\begin{equation}\label{cond2}
{\<_Y\com\rel{F}(R)\com\<_X}\;=\;{\<_Y\com\rel{F}(R)},
\end{equation}
so that the coalgebraic simulations for a right-stable order $\<$ can be
equivalently defined by means of the asymmetric definition on the
right-hand side of equality (\ref{cond2}). If the order $\<$ can be used only on
one of the sides of the definition, the verification of the properties of the
induced coalgebraic simulations becomes much easier than when using the original
definition.

It was quite surprising to discover that the easiest way to prove the
properties of the ``simulated by'' relations which come from symmetric
properties such as composition-preserving (that are also satisfied by the
corresponding inverse relations ``simulates'') is to break
that symmetry by considering the asymmetric definition of coalgebraic
simulations that only use $\<_Y$; certainly, this is only possible when the
defining order $\<$ is right-stable.

Stability is used in \cite[Lemma 5.3]{HughesJacobs04} to prove that lax relation
lifting preserves composition of relations, which is needed
to prove \cite[Lemma 5.4(2)]{HughesJacobs04}, the crucial fact that the induced
similarity relation is transitive; this need not be the case for the
simulation notion defined by an arbitrary order $\<$.

\section{On stability of simulation and anti-simulation}

Plain simulations between labeled transition systems can be defined as
coalgebraic simulations considering the functor $F=\ps^A$ ($G^A$
denote the funtor $X\mapsto (G(X))^A$) with the order
$\<$ given by $\alpha\<\beta$ for $\alpha,\beta:A\lra\ps X$ iff for all $a\in
A,\;$ $\alpha(a)\subseteq\beta(a)$.

\begin{lemma}
The order $\<$ defining plain simulations for labeled transition systems is
right-stable.
\end{lemma}

\begin{corollary}\label{cor1}
Plain simulations between labeled transition systems can be defined as the
$(\<_Y\com\rel{F})$-coalgebras.
\end{corollary}

It is worth examining the consequences of the removal of $\<_X$
from the original definition of coalgebraic simulations in this particular case.
Both $\<_X$ and $\<_Y$ correspond to the inclusion order, but when applied at
the right-hand side it means that we can reduce the set of successors of the
simulating process $q$ when simulating the execution of $a$ by $p$. This means
that starting from a set $Y'\subseteq Y$ we can obtain an adequate subset
$Y''\subseteq Y'$. Instead, the application of $\<_X$ at the left-hand
side allows to enlarge the set of successors of the simulated process $p$
and this produces a set $X''$ larger than the given $X'$: one could say that we
need to consider ``new'' information not in $X'$, while  going from $Y'$ to
$Y''$ just ``removes'' some known information.

Another interesting point arises from the fact that every use of $\<_X$ at
the left-hand side can be ``compensated'' by removing at $Y$ the added states
and this is why Corollary~\ref{cor1} was correct, because we can
always avoid the introduction of new successors in the simulated process
by simply removing them at the right-hand side. However, the opposite procedure,
to compensate the removal of states by adding them at the simulated process side
is not always possible, since in general $X$ could be not big enough.

The anti-simulations can be defined as coalgebraic simulations by taking the
reversed inclusion order instead of $\subseteq$. It is
interesting to note that it is not
right-stable as the following counterexample shows. Let $X=\{x\}$ and
$Y=\{y_1,y_2\}$ and let $f:X\lra Y$ be such that $f(x)=y_1$. With these
definitions the pair $(Y,X)\in (id\times \ps f)^{-1}(\supseteq)$, since
$Y\supseteq \{y_1\}=\ps f(X)$, but it is obvious that there is no $A\subseteq X$
such that $Y=f(A)$ because $f$ is not surjective.

However, the order defining anti-simulations is stable as a consequence of the
following general result.

\begin{lemma}\label{stable-op}
$F$ with an order $\<$ is stable iff it is stable with the order
${\<^{\mathit{op}}}$.
\end{lemma}

\begin{proof}
It is shown in \cite[Lemma 4.2(4)]{HughesJacobs04} that
$\relord{F}{\<^{\mathit{op}}}(R)=(\relord{F}{\<}(R^{\mathit{op}}))^{\mathit{op}
}
$.
Then, on the one hand,
\begin{eqnarray*}
(Ff\times Fg)^{-1}(\relord{F}{\<^{\mathit{op}}}(R))
&=&(Ff\times Fg)^{-1}(\relord{F}{\<}(R^{\mathit{op}}))^{\mathit{op}}\\
&=& ((Fg\times Ff)^{-1}\relord{F}{\<}(R^{\mathit{op}}))^{\mathit{op}},
\end{eqnarray*}
and on the other hand,
\begin{eqnarray*}
\relord{F}{\<^{\mathit{op}}}((f\times g)^{-1}(R))
&=& (\relord{F}{\<}((f\times g)^{-1}(R))^{\mathit{op}})^{\mathit{op}}\\
&=&(\relord{F}{\<}((g\times f)^{-1}(R^{\mathit{op}})))^{\mathit{op}}.
\end{eqnarray*}

Since $R^{\mathit{op}}\subseteq Y\times X$ is a relation whenever
$R\subseteq X\times Y$ is so, and $f$, $g$, and $R$ are arbitrary,
we have shown that
\[
\relord{F}{\<}((f\times g)^{-1}(R))=(Ff\times Fg)^{-1}(\relord{F}{\<}(R))
\]
if and only if
\[
\relord{F}{\<^{\mathit{op}}}((f\times g)^{-1}(R))=(Ff\times
Fg)^{-1}(\relord{F}{\<^{\mathit{op}}}(R)),
\]
and therefore $F$ is stable for $\<$ iff it is stable for $\<^{\mathit{op}}$.
\qed
\end{proof}

\begin{corollary}\label{antisim}
The order $\<^\mathit{op}$ defining anti-simulations for transition
systems as coalgebraic simulations is stable.
\end{corollary}

One could conclude from the observation above that there is indeed a
natural argument supporting plain similarity as a ``right'' coalgebraic
similarity, definable by a right-stable order. This criterion could be adopted
to define right coalgebraic simulations, which plain similarity would
satisfy while the opposite relation ``is
simulated by'' would not. However, we immediately noticed that we
could define ``left-stable'' orders by interchanging
the roles of $Ff$ and $id$ in the definition of right-stable order, obtaining
the 
inverse inclusion in (\ref{cond1}). 

\begin{definition}
We will say that a functor $F$ with order $\<$ is \textbf{left-stable} if, for
every function $f:X\lra Y$, we have
\begin{equation}\label{left}
{(Ff\times id)^{-1}\<_Y}\;\subseteq\; \coprod_{id\times Ff}\<_X.
\end{equation}
\end{definition}

It is inmediate to check that an order $\<$ is left-stable iff the inverse
order $\<^\mathit{op}$ is right-stable. Moreover, left-stable orders have the
same structural properties that right-stable ones so that, in particular, they
are also stable and hence composition-preserving. 
But in this case it would be the inverse simulations, corresponding
to the ``is simulated by'' notion, that would be natural instead of plain
simulations.
As a conclusion, we could use right or left-stability as a criterion to choose a
natural direction for the simulation order. But the important
fact in both cases is that the simplified asymmetric definitions (using either
$\<_X$ or $\<_Y$) of coalgebraic simulations are much easier to handle than the
symmetric original definition (where both $\<_X$ and $\<_Y$ have to be used).

\section{Covariant-contravariant simulations and conformance simulations}

Covariant-contravariant simulations are defined by combining the conditions ``to
simulate'' and ``be simulated by'', using a partition of the
alphabet of actions of the compared labeled transition systems.

\begin{definition}
Given $c:X\lra\ps(X)^{Act}$ and $d:Y\lra\ps(Y)^{Act}$ labeled
transition systems for the alphabet $Act$, and $\{Act^r,Act^l,
Act^{\mathit{bi}}\}$ a partition of this alphabet, a
\textbf{$(Act^r,Act^l)$-simulation}
between $c$ and $d$ is a relation $S\subseteq X\times Y$ such
that for every $(x,y)\in S$ we have:
\begin{itemize}
\item for all $a\in Act^r\cup Act^{\mathit{bi}}$ and all
$x\stackrel{a}{\lra}x'$ there exists $y\stackrel{a}{\lra}y'$ with
$(x',y')\in S$.

\item for all $a\in Act^l\cup Act^{\mathit{bi}}$, and all
$y\stackrel{a}{\lra}y'$ there exists $x\stackrel{a}{\lra}x'$ with
$(x',y')\in S$.
\end{itemize}

We write $x\B y$, and say that $x$ is $(Act^r,Act^l)$-simulated by $y$, if
and only if there exists some $(Act^r,Act^l)$-simulation $S$ with $xSy$.
\end{definition}

A very interesting application of this kind of simulations is related with the
definition of adequate simulation notions for input/output (I/O) automata
\cite{LynchVaandrager87}. The classic approach to simulations is based on the
definition of semantics for reactive systems, where all the actions of the
processes correspond to input actions that the user must trigger.
Instead, whenever we have explicit output actions the situation is the
opposite: it is the system that produces the actions and the user who is
forced to accept the produced output. Then, it is natural to conclude that in
the simulation framework we have to dualize the simulation condition when
considering output actions, and this is exactly what our anti-simulation
relations do.

Covariant-contravariant simulations can be easily obtained as coalgebraic
simulations, as the following proposition proves.

\begin{proposition}\label{prop1}
$(Act^r,Act^l)$-simulations can be defined as the coalgebraic simulations for
the functor $F=\ps^{\mathit{Act}}$ with functorial
order $\Ord$ where, for each set $X$ and $\alpha,\alpha':Act\lra\ps(X)$,
we have $\alpha\Ord\alpha'$ if:
\begin{itemize}
\item for all $a\in Act^r\cup Act^{\mathit{bi}},\;$
$\alpha(a)\subseteq\alpha'(a)$, and

\item for all $a\in Act^l\cup Act^{\mathit{bi}},\;$
$\alpha(a)\supseteq\alpha'(a)$.
\end{itemize}
Note that in particular we have $\alpha(a)=\alpha'(a)$ for all $a\in
Act^{\mathit{bi}}$.
\end{proposition}

\begin{proof}
Intuitively, using the order $\Ord$ on
the left-hand side of $\relord{F}{\<}(R)$ allows us to remove $a'$-transitions
when $a'\in Act^l$, whereas using it on the right-hand
side of $\relord{F}{\<}(R)$ allows us to remove $a$-transitions when $a\in
Act^r$.

Let us suppose that we have a classic covariant-contravariant simulation $\B$
between labeled transition systems $c:P\lra\ps(P)^\mathit{Act}$ and
$d:Q\lra\ps(Q)^\mathit{Act}$ defined by $c(p)(a)= \{p' \mid
p\stackrel{a}{\lra}p'\}$ and
$d(q)(a)= \{q' \mid q\stackrel{a}{\lra}q'\}$.
We must show that if $p\B q$ then there exist
$p^{*}$ and $q^{*}$ such that
\begin{equation}\label{relacion hey}
c(p)\Ord p^{*}\rel{\ps^{\mathit{Act}}}(\B)q^{*}\Ord d(q).
\end{equation}
We define $p^*$ and $q^*$ as follows:
\begin{itemize}
\item $p^*$ has the same transitions as $c(p)$, except for those
transitions
$p\stackrel{a'}{\lra}p'$ with $a'\in Act^l$ such that
there is no $q'$ with $q\stackrel{a'}{\lra} q'$ and $p'\B q'$.

\item $q^*$ has the same transitions as $d(q)$, except for those
transitions
$q\stackrel{a}{\lra}q'$ with $a\in Act^r$ such that
there is no $p'$ with $p\stackrel{a}{\lra} p'$ and $p'\B q'$.
\end{itemize}

It is immediate from these definitions that $c(p)\Ord p^*$ and $q^*
\Ord d(q)$,
so we are left with checking that $p^*\rel{\ps^{\mathit{Act}}} q^*$.

Let $p'\in p^*(a)$ with $a\in Act^r$. By construction
of $p^*$, since we have not dropped any $a$-transitions from $p^*$,
$p\stackrel{a}{\lra}p'$.
Using the fact that $\B$ is a classic covariant-contravariant
simulation, there exists $q'$ such that $q\stackrel{a}{\lra}q'$ with
$p'\B q'$,
and, again by construction, $q'\in q^*(a)$ because there is some
$p\stackrel{a}{\lra}p'$ with $p'\B q'$.
Similarly, if $p'\in p^*(a)$ with $a'\in Act^l$, by
construction of
$p^*$ there must exist some $q'$ such that $q\stackrel{a'}{\lra}q'$ with
$p'\B q'$.
Again, since we have not removed any $a'$-transitions from $d(q)$ in
$q^*$,
it must be true that $q'\in q^*(a)$.
Finally, if $p'\in p^*(a)$ with $a\in Act^\mathit{bi}$ we have that
$p\stackrel{a}{\lra}p'$ and hence there exists $q'$ such that
$q\stackrel{a}{\lra}q'$ with $p'\B q'$, but also
$q'\in q^*(a)$.

The argument that shows that for every $q'\in q^*(a)$ there exists some
$p'\in p^*(a)$ with $p' \B q'$ is analogous.

We show now the other implication, that a coalgebraic covariant-contravariant
simulation is a classic one.
In this case we start from coalgebras $c$ and $d$ that satisfy
relation~(\ref{relacion hey}) whenever $p\B q$.

If $p\stackrel{a}{\lra}p'$ for $a\in Act^r$, then $p'\in p^*(a)$
because $c(p) \Ord p^*$ and, since $p^*\rel{\ps^\mathit{Act}}(\B)q^*$, there is
some
$q'\in q^*(a)$ with $p'\B q'$. Again, the definition of $\Ord$ ensures
that $q^*(a)\subseteq d(q)(a)$ and hence $q\stackrel{a}{\lra}q'$ as
required.
Similarly, if $q\stackrel{a'}{\lra}q'$ for $a'\in Act^l$, then
$q'\in q^*(a)$ because $q^*\Ord d(q)$ and thus, as in the previous
case, there
exists $p'\in p^*(a)$ with $p'\B q'$ and $p\stackrel{a'}{\lra}p'$.
Finally if $p\stackrel{a}{\lra}p'$ for $a\in Act^\mathit{bi}$
(resp. $q\stackrel{a}{\lra}q'$), again by the definition of $\Ord$ we
have
$p'\in p^*(a)$ (resp. $q'\in q^*(a)$) and, from
$p^*\rel{\ps^\mathit{Act}}(\B)q^*$, it follows that there exists $q'\in q^*(a)$
(resp. $p'\in p^*(a)$) such that
$p'\B q'$; by the definition of $\Ord$,  $q\stackrel{a}{\lra}q'$ (resp.
$p\stackrel{a}{\lra}p'$).
\qed
\end{proof}

The other new kind of simulations in which we are interested is that of
conformance simulations, where the conformance relation in
\cite{Leduc92,Tretmans96} meets the simulation world in a nice way. In the
definition below we will write $p\stackrel{a}{\lra}$ if $p\stackrel{a}{\lra}p'$
for some $p'$.

\begin{definition}
Given $c:X\lra\ps(X)^A$ and $d:Y\lra\ps(Y)^A$ two labeled
transition systems for the alphabet $A$, a \textbf{conformance
simulation} between them is a relation $R\subseteq X\times Y$ such that whenever
$pRq$, then:
\begin{itemize}
\item For all $a\in A$,
if $p\stackrel{a}{\lra}$ we must also have
$q\stackrel{a}{\lra}$ (this means, using the usual notation for process
algebras,
that $I(p)\subseteq I(q)$).

\item For all $a\in A$ such that $q\stackrel{a} {\lra}q'$ and
$p\stackrel{a}{\lra}$, there exists some $p'$ with
$p\stackrel{a}{\lra}p'$ and $p'R q'$.
\end{itemize}
\end{definition}

Conformance simulations allow the extension of the set of actions
offered by a process, so that in particular we will have $a<a+b$, but they
also consider that a process can be ``improved'' by reducing the
nondeterminism in it, so that $ap+aq<ap$. In this way we have again
a kind of covariant-contravariant simulation, not driven by the alphabet of
actions executed by the processes but by their nondeterminism.

Once again, conformance simulations can be defined as coalgebraic simulations
taking the adequate order on the functor defining labeled transition systems.

\begin{proposition}
Conformance simulations can be obtained as the coalgebraic
simulations for the order $\CSO$ on the functor $\ps^A$, where for any set $X$
we have $u\CSO_X v$ if for every $u,v:A\lra\ps X$ and $a\in A$:
\begin{itemize}
\item either $u(a)=\emptyset$, or

\item $u(a)\supseteq v(a)$ and $v(a)\neq\emptyset$.
\end{itemize}
\end{proposition}

\begin{proof}
Let us first prove that $\CSO_X$ is indeed an order. It is clear that the only
not immediate property is transitivity. To check it, let us take
$u\CSO_X v\CSO_Y w$: if $u(a)=\emptyset$ we are done; otherwise, we have
$u(a)\supseteq v(a)$ and $v(a)\neq\emptyset$, so that we also have
$v(a)\supseteq w(a)$ and $w(a)\neq\emptyset$, obtaining $u(a)\supseteq
w(a)$ and $w(a)\neq\emptyset$.

Now, we can interpret that using the order $\CSO$ on
the left-hand side of $\relord{F}{\<}(R)$ allows us to remove all
$a$-transitions except for the last one, whereas using it on the
right-hand side allows us to remove all $b$-transitions for
$b\in B$, where $B$ is any set of actions. 
But again, as in the proof of Proposition~\ref{prop1}, we can compensate these
additions with the corresponding removals at the other side and the proof
follows in an analogous way.
\qed
\end{proof}

Next we check that the order $\Ord$ defining covariant-contravariant
simulations is stable.

\begin{lemma}\label{lemma3}
Given a partition $\{Act^r,Act^l,Act^{\mathit{bi}}\}$ of $Act$
the order $\Ord$ for the functor $\ps^{\mathit{Act}}$  defining
covariant-contravariant simulations for transition systems is stable.
\end{lemma}
\begin{proof}
It is clear that the order $\Ord$ can be obtained as the
product of a family of orders $\<^a$ for the functor $\ps$, with $a\in Act$.
This is indeed the case taking ${\<^a_X}={\subseteq_X}$ for $a\in
Act^r$, ${\<^a_X}={\supseteq_X}$ for $a\in Act^l$ and ${\<^a_X}={=_X}$ for $a\in
Act^{\mathit{bi}}$. Then it is easy to see that to obtain that $\Ord$ is stable
it is enough to prove that each of the orders $\<^a$ is stable.

This latter requirement is straightforward because, for $a\in Act^r$, $\<^a$ is
right-stable; for $a\in Act^l$ the order $\<^a$ is left-stable; and for $a\in
Act^{\mathit{bi}}$, $\<^a$ is the equality relation, which is both right and
left-stable, for every functor $F$. \qed
\end{proof}

Certainly, the order defining covariant-contravariant simulations is not
right-stable nor left-stable, but in the
proof above we have used the power of these two properties thanks to the fact
that the order $\Ord$ can be factorised as the product of a family of orders
that are either right-stable or
left-stable. Then we can obtain the following sequence of general definitions
and results, from which Lemma~\ref{lemma3} could be
obtained as a simple particular case.\footnote{Instead of removing the above,
we have preferred to maintain the sequence of results in the order in which we
got them, starting with our motivating example.}

\begin{definition}
We say that an order $\<$ on a functor $F^A$ is \textbf{action-distributive} if
there
is a family of orders $\<^a$ on $F$ such that
\[
f\< g \iff \textrm{$f(a)\<^a g(a)$ for all $a\in A$}.
\]
Whenever $\<$ can be distributed in this way we will write $\<\;=\prod_{a\in
A}\<^a$.
\end{definition}

\begin{definition}
We say that an action-distributive order $\<$ on $F^A$ is
\textbf{side stable} if for the decomposition  ${\<}={\prod_{a\in A}\<^a}$ we
have that each order $\<^a$ is either right-stable or left-stable.

By separating the right-stable and the left-stable components we obtain ${\<}={\<^l\times\<^r}$, where $A^r$ (resp. $A^l$) collects the set of
arguments\footnote{We have assumed here a partition $\{A^l, A^r\}$ of the set $A$ into two sets of right-stable and left-stable components.
Obviously, if there were some arguments $a\in A$ on which $\<^a$ is both right-stable and left-stable then the decomposition would not be
unique, but the result would be valid for any such decomposition.}
$a\in A$ with $\<^a$ right-stable (resp. left-stable). We extend $\<^l$ and
$\<^r$ to obtain a pair of orders on $F^A$, $\<^{\bar{l}}$ and $\<^{\bar{r}}$, defined by:
\begin{itemize}
\item $f\<^{\bar{r}} g$ iff $f(a)\<^a g(a)$ for all $a\in A^r$ and $f(a)=g(a)$
for all $a\in A^l$.

\item $f\<^{\bar{l}} g$ iff $f(a)\<^a g(a)$ for all $a\in A^l$ and $f(a)=g(a)$
for all $a\in A^r$.
\end{itemize}
\end{definition}

\begin{proposition}
The order $\<^{\bar{l}}$ is left-stable, while $\<^{\bar{r}}$ is right-stable. We have ${\<}= {(\<^{\bar{l}}\com \<^{\bar{r}})} =
{(\<^{\bar{r}}\com \<^{\bar{l}})}$, and therefore we also have $\<\; = (\<^{\bar{l}}\cup \<^{\bar{r}})^{*}$.
\end{proposition}

\begin{proposition}\label{prop-rl}
For any side stable order $\<$ on $F^A$, if we have a decomposition ${\<}
= {\<^l\times\<^r}$ based on a partition of $A$ into a set of
right-stable components $A^r$ and another set of left-stable components $A^l$,
then we can obtain the coalgebraic simulations for $\<$ as the
$(\<^{\bar{r}}_Y\com\rel{F}\com\<^{\bar{l}}_X)$-coalgebras.
\end{proposition}

\begin{proof}
By definition,
$\relord{F}{\<}(R)=\;\<_Y\com\rel{F}(R)\com\<_X$. Since
${\<}= {(\<^{\bar{r}}\com\<^{\bar{l}})} = {(\<^{\bar{l}}\com\<^{\bar{r}})}$,
we have:
\[
\begin{array}{rcl}
\<_Y\com\rel{F}(R)\com\<_X
&=&(\<^{\bar{l}}_Y\com\<^{\bar{r}}_Y)\com\rel{F}(R)\com(\<^{\bar{r}}_X\com\<^{\bar{l}}_X)\\
&=&\<^{\bar{l}}_Y\com (\<^{\bar{r}}_Y\com\rel{F}(R)\com\<^{\bar{r}}_X)\com\<^{\bar{l}}_X\\
&=&(\<^{\bar{l}}_Y\com\<^{\bar{r}}_Y)\com\rel{F}(R)\com\<^{\bar{l}}_X
\qquad\textrm{(by right-stability of $\<^{\bar{r}}$)}\\
&=&\<^{\bar{r}}_Y\com(\<^{\bar{l}}_Y\com\rel{F}(R)\com\<^{\bar{l}}_X)
\qquad\textrm{(since $\<^{\bar{r}}$ and $\<^{\bar{l}}$ commute)}\\
&=& \<^{\bar{r}}_Y\com\rel{F}(R)\com\<^{\bar{l}}_X
\qquad\textrm{(by left-stability of $\<^{\bar{l}}$)}\\
\noalign{\qed}
\end{array}
\]
\end{proof}

The characterization above still requires the use of the order on both sides
of the $\rel{F}(R)$ operator. However, the fact that $\<^{\bar{r}}_Y$ (resp.
$\<^{\bar{l}}_X$) is right-stable (resp. left-stable) makes the application
of this decomposition as simple as when coping with either a right or
left-stable order. 

\begin{proposition}
If ${\<}={\prod_{a\in A}\<^a}$ and $\<^a$ is stable for all $a\in A$, then $\<$
is
stable.
\end{proposition}
\begin{proof}
The result follows from the following chain of implications:
\[
\begin{array}{c@{\hskip .3cm}l}
& (u,v)\in (Ff\times Fg)^{-1}\relord{F}{\<}(R)\\
\Longleftrightarrow & Ff(u)\< z'\rel{F}(R) w'\<Fg(v)\\
\Longleftrightarrow & Ff(u)(a)\<^a
z'(a) \rel{F^a}(R) w'(a)\<^a Fg(v)(a),\;\textrm{for all $a$}\\
\Longleftrightarrow & (u(a),v(a))\in (Ff\times
Fg)^{-1}\relord{F}{\<^a}(R),\;\textrm{for all $a$}\\
\Longrightarrow & (u(a),v(a))\in \relord{F}{\<^a}((f\times
g)^{-1}R),\;\textrm{for all $a$}\\
\Longleftrightarrow & u(a)\<^a x'(a)\rel{F}((f\times g)^{-1}R)
y'(a)\<^a v(a),\;\textrm{for all $a$}\\
\Longleftrightarrow & (u,v)\in \relord{F}{\<}((f\times g)^{-1}R) \\
\end{array}
\]
\vspace{-1.3cm} 

\phantom{a}\qed
\end{proof}

\begin{corollary}
Any side stable order is stable.
\end{corollary}

\begin{corollary}
The order $\Ord$ defining covariant-contravariant simulations is side stable
and therefore it is stable too.
\end{corollary}

Next we consider the case of conformance simulations, for which we can obtain
similar results to those proved for covariant-contravariant simulations.

\begin{lemma}
The order $\CSO$ defining conformance simulations for transition systems
is stable.
\end{lemma}

\begin{proof}
Let $R\subseteq Z\times W$ be a relation and $f : X\lra Z$, $g: Y\lra W$
arbitrary functions. If $(u, v) \in
(\ps^Af\times\ps^Ag)^{-1}(\relord{\ps^A}{\CSO}(R))$, then there exist $z$ and
$w$ such that
\begin{equation}\label{cond-proof}
\ps^Af(u) \CSO z \mathrel{\rel{\ps^A}(R)} w \CSO \ps^Ag(v).
\end{equation}
We have to show that
$(u,v) \in \relord{\ps^A}{\CSO}((f\times g)^{-1}(R))$,
that is, there exist $x$ and $y$ such that
\[
u \CSO x \mathrel{\rel{\ps^A}((f\times g)^{-1}(R))} y \CSO v.
\]

Let us define $x : A\lra \ps(X)$ by $x(a) = u(a) \cap f^{-1}(z(a))$ and
$y : A\lra \ps(Y)$ by $y(a) = g^{-1}(w(a))$.
Then we have:
\begin{enumerate}
\item $u\CSO x$.

  If $u(a)=\emptyset$, there is nothing to prove.
  Otherwise, since $\ps^Af(u)\CSO z$ and $f(u(a))\neq\emptyset$, we have
  $f(u(a))\supseteq z(a)\neq \emptyset$ and hence
  $u(a) \supseteq u(a) \cap f^{-1}(z(a)) = x(a)\neq \emptyset$.
\item $y\CSO v$.

  If $w(a) = \emptyset$, then $y(a) = g^{-1}(w(a)) = \emptyset$.
  Otherwise, since $w\CSO \ps^Ag(v)$, we have  $w(a)\supseteq
g(v(a))\neq\emptyset$,
  so that $v(a)\neq\emptyset$ and
  $y(a) = g^{-1}(w(a)) \supseteq g^{-1}(g(v(a)))\supseteq v(a)$.
\item $x\, \mathrel{\rel{\ps^A}((f\times g)^{-1}(R))}\, y$.

  For every $a\in A$ we need to show that
  $x(a) \mathrel{\rel{\ps}((f\times g)^{-1}(R))} y(a)$, which means:
  \begin{enumerate}
  \item for every $p\in x(a)$ there exists $q\in y(a)$ such that
    $p \mathrel{(f\times g)^{-1}(R)} q$, that is, $f(p) R g(q)$; and
  \item for every $q\in y(a)$ there exists $p\in x(a)$ such that
    $p \mathrel{(f\times g)^{-1}(R)} q$, that is, $f(p) R g(q)$.
  \end{enumerate}
  In the first case, let $p\in x(a)$; by definition of $x$,
  $f(p)\in z(a)$.
  Now, from $z \mathrel{\rel{\ps^A}(R)} w$ we obtain that
  for each $p'\in z(a)$ there exists $q'\in w(a)$ such that
  $p' R q'$.
  Then, for $f(p)\in z(a)$ there exists $q'\in w(a)$ with $f(p) R q'$;
  and by definition of $y$, there exists $q\in y(a)$ with $q'=g(q)$ as
  required.

  In the second case, let $q\in y(a)$ so that $g(q)\in w(a)$.
  Again, from $z \mathrel{\rel{\ps^A}(R)} w$ it follows that there is
  $p'\in z(a)$ with $p' R g(q)$.
  Now, $f(u(a)) \supseteq z(a)$ because $u \CSO z$, so there exists
  $p\in u(a) \cap f^{-1}(z(a))$ with $f(p) = p'$, as required.
\qed
\end{enumerate}
\end{proof}

As in the case of covariant-contravariant simulations, conformance simulations
cannot be defined as coalgebraic simulations using neither a right-stable order
nor a left-stable order. But we can find in the arguments above the
basis for a decomposition of the involved order $\CSO$, according to the two
cases in its definition. Once again $\CSO$ is an action-distributive order on
$\ps^A$, but in order to obtain the adequate decomposition of $\CSO$ now
we also need to decompose the component orders $\<^a$.

\begin{definition}
We define the conformance orders\textbf{ $\CSNotEmpty$}, \textbf{$\CSEmpty$},
and \textbf{$\<^C$} on the
functor $\ps$ by:
\begin{itemize}
\item $x_1\CSEmpty x_2$ if $x_1=\emptyset$ or $x_1 = x_2$.

\item $x_1\CSNotEmpty x_2$ if $x_1\supseteq x_2$ and $x_2\neq\emptyset$,
 or $x_1 = x_2$.

\item $x_1\<^C x_2$ if $x_1\CSNotEmpty x_2$ or $x_1\CSEmpty x_2$.
\end{itemize}
\end{definition}

\begin{proposition}
The two relations $\CSEmpty$ and $\CSNotEmpty$ commute with each other:
\[(\CSEmpty\com \CSNotEmpty)=(\CSNotEmpty\com \CSEmpty),
\]
from where it follows that
$(\CSEmpty\cup\CSNotEmpty)^*\;=\;(\CSEmpty\com\CSNotEmpty)\;=\;(\CSNotEmpty\com
\CSEmpty)$. We also have $\<^C=(\CSEmpty\com \CSNotEmpty)$, from where
we conclude that $\<^C$ is indeed an order relation.
\end{proposition}

\begin{proof}
Let $u \mathrel{(\CSEmpty\com \CSNotEmpty)} v$: there is some $w$ such that $u
\CSNotEmpty w$ and $w \CSEmpty v$.
We need to find $w'$ such that $u\CSEmpty w'$ and $w' \CSNotEmpty v$.
If $w=\emptyset$ then it must be $u = \emptyset$ too, and we can take
$w'= v$; otherwise, it must be $v = w$ and we can take $w'= u$. The other
inclusion is similar. \qed
\end{proof}

\begin{corollary}
The order $\CSO$ defining conformance simulations can be decomposed
into $\prod_{a\in A}\<^a$ where, for each $a\in A$, we have
${\<^a}={\<^C}$ as defined above. Then, $\CSO=\prod_{a\in
A}(\CSNotEmptya \cup\CSEmptya)^*=\prod_{a\in A}(\CSNotEmptya) \com\prod_{a\in
A}(\CSEmptya)=\prod_{a\in A}(\CSEmptya) \com\prod_{a\in A}(\CSNotEmptya)$, so
that we obtain $\CSO$ as the composition of a right-stable order and a
left-stable order that commute with each other.
\end{corollary}

\begin{proposition}
For any pair of right (resp. left)-stable orders $\<^1$, $\<^2$ on $F$, their
composition also defines a right (resp. left)-stable order on $F$.
\end{proposition}
\begin{proof}
Given $f:X\lra Y$ we must show that
\[(id\times Ff)^{-1}(\<^1_Y\com\<^2_Y) \mathrel{\subseteq}\coprod_{(Ff\times
id)}(\<^1_X\com\<^2_X).\]

Let us assume that $(y,x)\in (id\times Ff)^{-1}(\<^1_Y\com\<^2_Y)$, that is,
$y\mathrel{(\<^1\com\<^2)} y'=Ff(x)$; then, there exists $y''\in FY$ such that
$y\<^2_Y y''$ and $y''\<^1_Y y'$. Graphically,
\begin{equation}\label{dia1}
\xymatrix@R=.8cm@C=.1cm{
{y} & {\<^2_Y}& {y''} & {\<^1_Y} & {y'}\\
& & & & {x}\ar@{|->}[u]_{Ff} }
\end{equation}
Since $\<^1_Y$ is right-stable we have that ${(id\times
Ff)^{-1}\<^1_Y}\subseteq{\coprod_{(Ff\times id)}\<^1_X}$. Hence, there
exists
$x''\in FX$ such that $Ff(x'')=y''$ and $x''\<^1_X x$, thus turning
diagram~(\ref{dia1}) into the following:
\begin{equation}\label{dia2}
\xymatrix@R=.8cm@C=.1cm{ {y} & {\<^2_Y}& {y''} & & \\
& & {x''}\ar@{|->}[u]_{Ff} & {\<^1_X}& {x} }
\end{equation}

Now, we can apply right-stability of $\<^2$: since we have
$(y,x'')\in {(id\times Ff)^{-1}\<^2_Y}\subseteq{\coprod_{(Ff\times id)}\<^2_X}$,
there exists $x'\in FX$ such that $Ff(x')=y$ and $x'\<^2_X x''$.
Thus, diagram (\ref{dia2}) becomes
\begin{equation}\label{dia3}
\xymatrix@R=.8cm@C=.1cm{
{y} & & & & \\
{x'}\ar@{|->}[u]_{Ff} &{\<^2_X} & {x''} & {\<^1_X}& {x} }
\end{equation}
which means that there exist $x',x''\in
FX$ such that $Ff(x')=y$, $x'\<^2_X x''$ and $x''\<^1_X x$, or equivalently,
that $(y,x)\in\coprod_{(Ff\times id)}(\<^1_X\com\<^2_X)$, as we had to prove.
\qed
\end{proof}

\begin{proposition}
If $\<^r$ is a right-stable order on $F$ and $\<^l$ is a left-stable order on
$F$ that commute with each other, then their composition defines a stable order
on $F$. Moreover, the coalgebraic simulations for the order
${\<}={\<^r\com\<^l}$
can be equivalently defined as the $(\<^r\com\rel{F}(R)\com\<^l)$-coalgebras.
\end{proposition}

\begin{proof}
Let $R\subseteq Z\times W$ be a relation, $f : X\lra Z$ and $g: Y\lra W$
arbitrary functions, and ${\<}={\<^r\com\<^l}$. Let us suppose
that $(u, v) \in (Ff\times Fg)^{-1}(\relord{F}{\<}(R))$. Then, since $\<^r$ and
$\<^l$ commute with each other, using Proposition~\ref{prop-rl}, there exist
$z'$, $w'$ such that
\begin{equation}\label{cond-proof1}
Ff(u) \<^{l}_Z z' \mathrel{\rel{F}(R)} w' \<^{r}_W Fg(v).
\end{equation}
If we write $z$ for $Ff(u)$ and $w$ for $Fg(v)$, then
equation (\ref{cond-proof1}) is equivalent to
\begin{equation}\label{diag-stable1}
\xymatrix@R=.8cm@C=.1cm{
{z} & {\<^{l}_Z}& {z'} & {\rel{F}(R)} & {w'} & {\<^{r}_W}& {w} \\
{u}\ar@{|->}[u]^{Ff}  & & & & & & {v}\ar@{|->}[u]_{Fg} }
\end{equation}
and we have to show that $(u,v) \in \relord{F}{\<}((f\times g)^{-1}(R))$, that
is,  that there exist $x$ and $y$ such that
\[
u \<^{l}_X x \mathrel{\rel{F}((f\times g)^{-1}(R))} y \<^{r}_Y v.
\]

Using that $\<^{r}$ is right-stable on the rhs of equation
(\ref{cond-proof1}), we get $(w',v)\in {(id\times
Fg)^{-1}\<^{r}_W}\subseteq{\coprod_{(Fg\times id)}\<^{r}_Y}$, so
that there is some
$y\in FY$ such that $Fg(y)=w'$, with $y\<^{r}_Y v$. Graphically, diagram
(\ref{diag-stable1}) becomes
\begin{equation}\label{diag-stable11}
\xymatrix@R=.8cm@C=.1cm{
{z} & {\<^{l}_Z}& {z'} & {\rel{F}(R)} & {w'} & &  \\
{u}\ar@{|->}[u]^{Ff}  & & & & {y}\ar@{|->}[u]_{Fg} & {\<^{r}_Y} & {v} }
\end{equation}

Analogously, applying the left-stability of order $\<^{l}_Z$ we get that
there is
some $x\in FX$ with $Ff(x)=z'$ such that $u\<^{l}_X x$. Or graphically,
\begin{equation}\label{diag-stable2}
\xymatrix@R=.8cm@C=.1cm{
& & {z'} & {\rel{F}(R)} & {w'} & & \\
{u} & {\<^{l}_X}& {x}\ar@{|->}[u]^{Ff} & & {y}\ar@{|->}[u]_{Fg} & {\<^{r}_Y}&
{v} }
\end{equation}
But diagram (\ref{diag-stable2}) is just what we had to prove, since we have
found $x,y$ such that $(x,y) \in (Ff\times Fg)^{-1}(\rel{F}(R))=
\rel{F}((f\times g)^{-1}(R))$ with $u\<^{l}_X x$,
$y\<^{r}_Y v$ or, in other words, $(u,v)\in\relord{F}{\<}((f\times
g)^{-1}(R))$.\qed
\end{proof}

In particular, for our running example of conformance simulations we obtain the
corresponding factorization of the definition of coalgebraic simulations for the
order $\CSO$:

\begin{corollary} 
Coalgebraic simulations for the conformance order $\CSO$ can
be equivalently
defined as the $(\prod_{a\in A}(\CSNotEmptya_Y)\com\rel{F}(R)\com \prod_{a\in
A}(\CSEmptya_X))$-coalgebras.
\end{corollary}

\section{Conclusion}

We have presented in this paper two new simulation orders induced by two
criteria that capture the difference between input and output
actions and the implementation notions that are formalized by the conformance
relations.

In order to apply the general theory of coalgebraic simulations to them, we
identified the corresponding orders on the functor defining labeled transition
systems. However, it was not immediate to prove
that the obtained orders had the desired good properties since the usual way to
do it, namely, by establishing stability as a consequence of a stronger property
that we have called right-stability, is not applicable in this case.

Trying to adapt that property to our situation we have discovered
several interesting consequences. We highlight the fact that
right-stability is an assymetric property which has proved to be very
useful for the study of a ``reversible'' concept such as that of relation, since
it is clear that any structural result on the theory of relations should remain
true when we reverse the relations, simply ``observing'' them ``from the other
side''. Two consequences of that assymetric approach followed: first we noticed
that we could use it to point the simulation orders in some natural
way; secondly we also noticed that by dualizing the right-stability condition we
could obtain left-stability.

But the crucial result in order to be able to manage more complicated
simulation notions, as proved to be the case for our new
covariant-contravariant simulations and the conformance simulations, was the
discovery of the fact that both of them could be factorized into the
composition of a right-stable and a left-stable component. Exploiting this
decomposition we have been able to easily adapt all the techniques that
had proved to be very useful for the case of right-stable orders.

We plan to expand our work here in two different directions. The first one is
concerned with
the two new simulated notions introduced in this paper: once we know that they
can be defined as stable coalgebraic simulations and therefore have all the
desired basic properties of simulations, we will continue with their study  by
integrating them into our unified presentation of the semantics for processes
\cite{DeFrutosEtAl08c}. Hence we expect to obtain, in particular, a clear
relation between conformance similarity and the classic similarity orders
as well as an algebraic characterization for the new semantics.
In addition, we plan to continue with our study of stability, which has proved
to be a crucial property in order to understand the notion of coalgebraic
simulation,
thus making it possible to apply the theory to other examples
like those studied in this paper.


\begin{thebibliography}{10}

\bibitem{AczelMendler89}
P.~Aczel and N.~P. Mendler.
\newblock A final coalgebra theorem.
\newblock In D.~H. Pitt, D.~E. Rydeheard, P.~Dybjer, A.~M. Pitts, and
  A.~Poign{\'e}, editors, {\em Category Theory and Computer Science}, vol
  389 of {\em LNCS}, pages 357--365. Springer,
  1989.

\bibitem{Bloom95}
B.~Bloom, S.~Istrail, and A.~R. Meyer.
\newblock Bisimulation can't be traced.
\newblock {\em J. ACM}, 42(1):232--268, 1995.

\bibitem{DeFrutosGregorio08}
D.~de~Frutos-Escrig and C.~Gregorio-Rodr\'{\i}guez.
\newblock Universal coinductive characterisations of process semantics.
\newblock In G.~Ausiello, J.~Karhum{\"a}ki, G.~Mauri, and C.-H.~L. Ong,
  editors, {\em IFIP TCS}, vol 273 of {\em IFIP}, pages 397--412. Springer,
  2008.

\bibitem{DeFrutosEtAl08c}
D.~de~Frutos~Escrig, C.~Gregorio-Rodr{\'\i}guez, and M.~Palomino.
\newblock On the unification of semantics for processes: observational
  semantics.
\newblock In M.~Nielsen, A.~Kucera, P.~Bro~Miltersen, C.~Palamidessi, P.~Tuma,
  and F.~Valencia, editors, {\em SOFSEM 09: Theory and Practice of Computer
  Science. 35th International Conference on Current Trends in Theory and
  Practice of Computer Science, Proceedings}, vol 5404 of {\em LNCS},
pages 279--290. Springer, 2009.

\bibitem{Groote92}
J.~F. Groote and F.~W. Vaandrager.
\newblock Structured operational semantics and bisimulation as a congruence.
\newblock {\em Inf. Comput.}, 100(2):202--260, 1992.

\bibitem{HughesJacobs04}
J.~Hughes and B.~Jacobs.
\newblock Simulations in coalgebra.
\newblock {\em TCS}, 327(1-2):71--108, 2004.

\bibitem{Jacobs03}
B.~Jacobs and J.~Hughes.
\newblock Simulations in coalgebra.
\newblock In H.~P. Gumm, editor, {\em CMCS'03: 6th International Workshop on
  Coalgebraic Methods in Computer Science}, volume~82, 2003.

\bibitem{LarsenSkou91}
K.~G. Larsen and A.~Skou.
\newblock Bisimulation through probabilistic testing.
\newblock {\em Inf. Comput.}, 94(1):1--28, 1991.

\bibitem{Leduc92}
G.~Leduc.
\newblock A framework based on implementation relations for implementing
  {LOTOS} specifications.
\newblock {\em Computer Networks and ISDN Systems}, 25(1):23--41, 1992.

\bibitem{LynchVaandrager87}
N.~A. Lynch and M.~R. Tuttle.
\newblock Hierarchical correctness proofs for distributed algorithms.
\newblock In {\em Sixth Annual ACM Symposium on Principles of Distributed
  Computing}, pages 137--151, 1987.

\bibitem{Park81}
D.~Park.
\newblock Concurrency and automata on infinite sequences.
\newblock In P.~Deussen, editor, {\em Theoretical Computer Science, 5th
  GI-Conference, Proceedings}, vol
  104 of {\em LNCS}, pages 167--183. Springer,
  1981.

\bibitem{Tretmans96}
J.~Tretmans.
\newblock Conformance testing with labelled transition systems: Implementation
  relations and test generation.
\newblock {\em Computer Networks and ISDN Systems}, 29(1):49--79, 1996.

\bibitem{VanGlabbeek01}
R.~J. van Glabbeek.
\newblock The linear time-branching time spectrum {I}: The semantics of
  concrete, sequential processes.
\newblock In J.~A. Bergstra, A.~Ponse, and S.~A. Smolka, editors, {\em Handbook
  of process algebra}, pages 3--99. 2001.

\end{thebibliography}
\end{document}